\documentclass[showpacs,pra,twocolumn]{revtex4}
\usepackage{graphicx}
\usepackage{amsmath}
\usepackage{amsthm}
\usepackage{amssymb}
\usepackage{bbm}
\usepackage{color}

\newcommand{\ket}[1]{\mbox{$ | #1 \rangle $}} 
\newcommand{\bra}[1]{\mbox{$ \langle #1 | $}}
\newcommand{\braket}[1]{\left<#1\right>}

\newcommand{\proj}[1]{\ket{#1}\bra{#1}}

\newcommand{\beq}{\begin{equation}}
\newcommand{\eeq}{\end{equation}}
\newcommand{\best}{\begin{equation*}}
\newcommand{\eest}{\end{equation*}}

\newcommand{\idmap}{{\rm id}}

\newcommand{\mR}{\mathbb{R}}

\DeclareMathOperator{\Tr}{tr}

\begin{document}

\newcommand{\bk}[1]{\ketbra{#1}{#1}}
\newcommand{\trace}{\ensuremath{\operatorname{tr}}}
\newcommand{\density}[1]{\ensuremath{\mathcal{D}(\mathcal{H}_{#1})}}
\newcommand{\map}{\ensuremath{\operatorname{\Lambda}}}
\newcommand{\mapp}{\ensuremath{\operatorname{\Theta}}}
\newcommand{\ie}{\emph{i.e.}}
\newcommand{\eg}{\emph{e.g.}}
\newcommand{\expect}[1]{\left<#1\right>}
\newcommand{\abs}[1]{\left| #1 \right|}
\newcommand{\var}[1]{\ensuremath{\text{Var}( #1 )}}
\newtheorem{proposition}{Proposition}[section]
\newtheorem{problem}{Problem}[section]
\newtheorem{definition}{Definition}[section]
\newtheorem{theorem}{Theorem}[section]

\newcommand{\C}{\ensuremath{\mathbbm C}}
\newcommand{\be}{\begin{equation}}
\newcommand{\ee}{\end{equation}}
\newcommand{\eea}{\end{eqnarray}}
\newcommand{\bea}{\begin{eqnarray}}
\newcommand{\bd}{\ensuremath{b^\dagger}}
\newcommand{\va}[1]{\ensuremath{(\Delta#1)^2}}
\newcommand{\vasq}[1]{\ensuremath{[\Delta#1]^2}}
\newcommand{\varho}[1]{\ensuremath{(\Delta_\rho #1)^2}}
\newcommand{\ex}[1]{\ensuremath{\left\langle{#1}\right\rangle}}
\newcommand{\exs}[1]{\ensuremath{\langle{#1}\rangle}}
\newcommand{\mean}[1]{\ensuremath{\langle{#1}\rangle}}
\newcommand{\exrho}[1]{\ensuremath{\left\langle{#1}\right\rangle}_{\rho}}
\newcommand{\eins}{\openone}
\newcommand{\WW}{\ensuremath{\mathcal{W}}}
\newcommand{\HH}{\ensuremath{\mathcal{H}}}
\newcommand{\PP}{\ensuremath{\mathcal{P}}}
\newcommand{\QQ}{\ensuremath{\mathcal{Q}}}
\newcommand{\ketbra}[1]{\ensuremath{| #1 \rangle \langle #1 |}}
\newcommand{\kommentar}[1]{}

\title{Entanglement verification with realistic measurement devices via squashing operations}

\author{Tobias Moroder$^{1,2,3}$, 
Otfried G\"uhne$^{4,5}$, Normand Beaudry$^{3}$, 
\\ Marco Piani$^{3}$, Norbert L\"utkenhaus$^{1,2,3}$} 

\affiliation{$^1$ Quantum Information Theory Group, Institute of 
Theoretical Physics I, University Erlangen-Nuremberg, Staudtstra{\ss}e 7/B2, 
91058 Erlangen, Germany \\
$^2$ Max Planck Institute for the Science of Light, G{\"u}nther-Scharowsky-Stra{\ss}e~1/24, 
91058 Erlangen, Germany \\
$^3$ Institute for Quantum Computing \& Department of Physics
  and Astronomy, University of Waterloo, 200 University 
Avenue West, N2L 3G1 Waterloo, Ontario, Canada \\
$^4$ Institut f\"ur Quantenoptik und Quanteninformation, 
\"Osterreichische Akademie der Wissenschaften, Technikerstra{\ss}e~21A, 
A-6020 Innsbruck, Austria \\
$^5$ Institut f\"ur Theoretische Physik, Universit\"at Innsbruck, 
Technikerstra{\ss}e~25, A-6020 Innsbruck, Austria}

\date{\today}
\begin{abstract}
Many protocols and experiments in quantum information science are described in terms of simple measurements on qubits. However, in a real implementation, the exact description is more difficult, and more complicated observables are used. The question arises whether a claim of entanglement in the simplified description still holds, if the difference between the realistic and simplified models is taken into account. We show that a positive entanglement statement remains valid if a certain positive linear map connecting the two descriptions---a so-called squashing operation---exists; then lower bounds on the amount of entanglement are also possible. We apply our results to polarization measurements of photons using only threshold detectors, and derive procedures under which multi-photon events can be neglected.
\end{abstract}

\maketitle

\section{Introduction}

According to Asher Peres, entanglement is ``a trick that quantum 
magicians use to produce phenomena that cannot be imitated by classical 
magicians'' \cite{bruss02a}. 
Because of the key role of entanglement in applications lots of 
effort is put into realizing this fragile resource in the lab, for 
example via parametric down-conversion (PDC) sources or with ion traps, 
to only name a few. In a real experiment it is of course desirable 
to unambiguously verify the creation of entanglement, and in fact many 
different operational tools have been developed over the past years to 
achieve this task, cf.~Ref.~\cite{toth08} for a review. A reliable 
entanglement verification has to satisfy a few crucial criteria 
\cite{vanenk07a}; most importantly the verification method should not 
rely on assumptions from the entanglement generation process, but 
instead on the information acquired about the system via measurements. 
Moreover the obtained data 
should be considered under a worst case scenario.
{That is,  in the spirit of 
Ref.~\cite{horrorjaynes}, the test is only 
considered to be affirmative if,  in the limiting case of an infinite 
number of experimental runs}, the data exclude compatibility with all 
separable states.  This viewpoint is even essential for certain 
tasks like quantum cryptography \cite{curty04a}. 

In any case, it is typical to allow one basic ingredient: since one usually
believes in quantum mechanics, it is common to assume that one is 
equipped with an accurate quantum mechanical description of the 
employed measurement devices; {the actual testing or the (experimental) characterization 
of} a measurement device is anyway often combined with other assumptions concerning the 
generated state \cite{mayers03,dmms07,lundeen09a}. Note that if one does not restrict oneself to this model then one can still use Bell inequalities for the verification.  This leads to the known drawback that the entanglement of certain states can never be verified \cite{werner_sep} 
and there is even the conjecture that complete classes of interesting
entangled states may be undetectable \cite{peresconjecture}. However 
this will not concern us here, and we always assume to have an operator 
set associated with the observed data, which allows us connect the data
to quantum mechanical  quantities.
 
An example of a straightforward and hence quite often applied 
entanglement verification method, \eg, Ref.~\cite{james01a}, 
is the following procedure which we call the
\emph{tomography entanglement test} in the following: 
Since the useful entanglement might be confined to a low 
dimensional subspace, \eg, the single photon-pair subspace of a 
PDC source or two very long-lived energy levels of two ions in a trap, one just performs a 
few different measurements to obtain tomography 
on this subspace. After several runs of the experiment 
one has collected enough data to reconstruct the underlying 
density operator on this subspace via some reconstruction 
technique. Note that here one employs the knowledge of the 
measurement description. In order to check for entanglement 
one just investigates whether this reconstructed density 
operator describes an entangled state or not. 

However, does one really verify entanglement via this method? 
The problem lies within the measurement description, because 
such ideal measurements, as the ones used in the reconstruction 
mechanism, might not have actually been performed in the experiment. A good example 
{is represented by} the polarization measurement with two threshold detectors (see 
also Fig.~\ref{fig:detectionsetup}), which is typically employed in photonic experiments. 
Apart from usually acting onto several input modes at once, this 
device does not even respond solely to the single photon 
subspace, since such detectors cannot resolve the number 
of photons. Hence the question arises whether one still 
verifies the entanglement if a more realistic measurement 
description is employed. It is the main purpose of this 
paper to study this question. Note that the aforementioned 
scenario often occurs, not because one is not aware of the 
more realistic model, but because an oversimplified measurement 
description is employed in order to ease the task of entanglement 
verification.  

Specific instances of the problems considered here have been
investigated in several works in the literature. In 
Ref.~\cite{seevinck-2007-76} inequalities for the detection 
of entanglement for two qubits have been proposed, where 
the measurement's devices can be misaligned to a certain degree.
Bell-type inequalities which are independent of the spectrum of the 
measured observables have been recently introduced in 
Ref.~\cite{vogel-2009}. Moreover, for an experiment with
photons from atomic ensembles, an entanglement verification scheme which 
takes multi-photon events into account has been introduced 
\cite{lougovski-2009} and implemented \cite{kimble09}.

In this paper we proceed along the following lines: 
In Section II we provide an example of a tomography entanglement 
test which indeed leads to the wrong conclusion about the presence 
of entanglement under a small, physical change of the employed 
measurement description. 

In Section III, we start to investigate 
under which conditions such mistakes can safely be excluded. 
In short, the entanglement verification process remains stable 
as soon as the considered set of operators are connected by a 
positive but not necessarily completely positive map, the so-called 
squashing operation. Similar 
relations between different measurement schemes have recently been
introduced in the context of QKD, 
cf.~Ref.~\cite{squash1,squash2}, and even other known verification 
methods can be cast into this framework. However, complete positivity
of the connection map was required there.

In Section IV  we reformulate the existence of 
such a positive map into a necessary and sufficient condition which 
provides a particular intuitive solution for the tomography 
entanglement test: The map exists if and only if each classical 
outcome pattern from the refined set of observables remains 
compatible with the oversimplified set of observables. 

Then, in Section V we prove that the aforementioned polarization measurement 
with threshold detectors along all three different polarization axes
represents an example which indeed can only be linked to its single
 photon realization by a positive but not completely positive map. 
This analysis concludes that the tomography entanglement test which 
is typically employed for a PDC source \cite{kwiat00a} or even in
multipartite photonic experiments \cite{wieczorek08a} (using the 
single photon assumption) can indeed be made error-free if the (local)
double click events are taken into account.

In Section VI we consider the issue of entanglement quantification, proving that one can in principle
get lower bounds on the entanglement of the physical state in terms of the entanglement of the operator that results from the local mapping
between the observables.

Finally, we conclude and provide an outlook on possible further directions. 

\section{An example for ion trap experiments} 

Let us first mention a simple, yet practically relevant example, 
which shows that the tomography entanglement test indeed can lead 
to a false conclusion about the presence of entanglement if the 
structure of the observables is not properly taken into account. 

For a single $^{40}$Ca-ion in a trap one typically considers only 
the lowest two energy levels given by a lower level $\ket{S}=\ket{1}$
and the upper level $\ket{D}=\ket{0}$ and treats them as the 
qubit \cite{haeffner-2008}. Resonance fluorescence provides a 
mechanism to read out the occupation number of the energy levels:
An electron in the $\ket{S}$ state is coupled to a higher energy 
level $\ket{P}$, and observing photons from the $\ket{S} \leftrightarrow \ket{P}$ 
transition signals that the qubit was in the state $\ket{S}$. 
This overall process corresponds to a projection onto the 
lower energy state and consequently allows to measure the 
$\sigma_z$ Pauli, while the measurement along different 
directions is achieved by a local basis rotations prior to 
the $\sigma_z$ measurement, cf. Ref.~\cite{haeffner-2008}. 

In order to avoid too many measurements it is common to 
measure the occupation probability only for the state $\ket{S}$, 
simply because for qubits the other probability equals 
$p(D)= 1-p(S)$ due to the normalization, and similar for 
the remaining basis settings. Suppose that one uses this measurement procedure to obtain tomography in order to verify the creation of entanglement between two separated ions in the trap. Consider now the example that the observed expectation values, abstractly denoted as $E_{ij}(p)$ and 
characterized by a noise parameter $p$ 
\footnote{Note that these expectation values 
can also be calculated from the state $\rho(p)$ of 
Eq.~(\ref{wernerstate}) by 
$\mathbbm{E}_{ij}(p)=\trace[\rho(p) F_i^{\rm A} \otimes F_j^{\rm B}]$ 
with $F^{\rm A}_i,F^{\rm B}_j \in \{ \ket{0}\bra{0}, \ket{x^+}\bra{x^+},\ket{y^+}\bra{y^+} \}$.}, 
may allow the reconstruction of the state 
\begin{equation}
\label{wernerstate}
\rho(p)=(1-p)\ketbra{\psi^+}+p\frac{\mathbbm{1}}{4},
\end{equation}
which is, by virtue of the PPT criterion, entangled for $p < 2/3$. 

However in practice the situation is more complicated since 
the ion is not a simple two-level system. To model this, 
one can add another energy level to only one of the ions, 
thereby enlarging the two-qubit system to a qubit-qutrit one. 
Without any additional information about the occupation number 
of this extra level, it is clear that the assignment $p(D)= 1-p(S)$ 
is not correct any more. Consequently the observed data $E_{ij}(p)$ 
can only verify entanglement for the case $p < 0.63$. This can be 
checked by using the tools from Ref.~\cite{curty07a}, in which 
the search for an appropriate separable state was phrased into 
a semidefinite program. Hence we have the interval $p \in [0.63, 2/3)$, 
for which the performed tomography 
entanglement test indicates the presence of entanglement although 
with the more realistic model it does not. Though this region might 
be small this error can become important in the multipartite scenario, 
where current experiments just operate at the border of genuine 
multipartite entanglement \cite{haeffner-2005-438,leibfriedsixghz,gao-2008}. 
Concerning the experimental consequences, however, two facts are important:

\begin{enumerate}

\item For experiments with ion traps it is known that the occupation 
probability for levels apart from the two logical states is very 
small, typically it is around $10^{-3}$ \footnote{Christian Roos, 
private communication.}. Given this additional measurement data, 
it is possible to provide a quantitative estimate of the resulting 
error in the used entanglement verification scheme, \eg, the mean 
value of an entanglement witness. For typical entanglement witnesses 
employed in those scenarios this error is far below the unavoidable 
statistical uncertainty, which is caused by the finite number of 
copies of a state available in any experiment.

\item Note that the probabilities $p(S)$ and $p(D)$ of each energy 
level can be measured independently by additional local rotations, 
hence at the expense of more measurements. Then the resulting 
probabilities correspond to the unnormalized two-level state 
$\rho_{\rm red}$ that is obtained from our modeled  
three-level system $\rho_{\rm tot}$ by a local projection, 
\ie, $\rho_{\rm red}= \Pi \rho_{\rm tot} \Pi,$ with 
$\Pi = \ketbra{S} + \ketbra{D}$. As long as we prove 
entanglement of the two-qubit system 
$\rho^{\rm AB}_{\rm red} =\mathbbm{1} \otimes \Pi \rho^{\rm AB}_{\rm tot} \mathbbm{1} \otimes \Pi  $, 
this also implies entanglement for the total state 
$\rho^{\rm AB}_{\rm tot}$, since the projection is local. 

For instance, if one measures a witness like
$
\WW
= 
\ketbra{00} + \ketbra{11} 
- \ketbra{x^+ x^+} - \ketbra{x^- x^-}
+ \ketbra{y^+ y^+} + \ketbra{y^- y^-},
$
with $\ket{x^{\pm}}=(\ket{0}\pm \ket{1})/\sqrt{2}$ and 
$\ket{y^{\pm}}=(\ket{0}\pm i\ket{1})/\sqrt{2}$ \cite{toth08}, 
the mean value of $\WW$ is nothing more than a linear combination of 
certain probabilities on the qubit space, and if the mean value 
is negative, the  state $\rho^{\rm AB}_{\rm red}$ and hence 
$\rho^{\rm AB}_{\rm tot}$ is entangled. 
This shows that additional dimensions of the Hilbert space 
alone do not invalidate the conclusion that the state is 
entangled when the measurement devices are 
characterized properly.

\end{enumerate}

\section{Positive squashing operations}

We are ready to formulate the problem that we solve 
throughout the subsequent sections. For each local measurement 
setup one has 
two different sets of ordered observables; a set of simple 
\emph{target} observables labeled as $T_i$ with $i=1,\dots,n$ 
acting on the Hilbert space $\mathcal{H}_T$ which are used for 
the entanglement verification process or in the reconstruction 
mechanism, and a different set of so-called \emph{full} operators 
denoted as $F_i$ with $i=1,\dots,n$ onto the Hilbert space 
$\mathcal{H}_F$ which represent the more realistic model of 
the actual observables in the experiment. In the above 
ion-trap example we considered the case of qubit target 
observables, while our full operators were acting on 
a qutrit system. 

Consider the case where in an experiment one measures 
the expectation values 
of the full operators $F_i$ but instead one interprets them 
as the expectation values of the target observables $T_i$. The question arises,
whether this may lead to a false entanglement verification. In the 
following we provide a simple condition on the two operator sets 
{alone that excludes such a possibility}, and hence guarantees the 
presence of entanglement. 

\begin{figure}
  \begin{center}
  \includegraphics[scale=0.3]{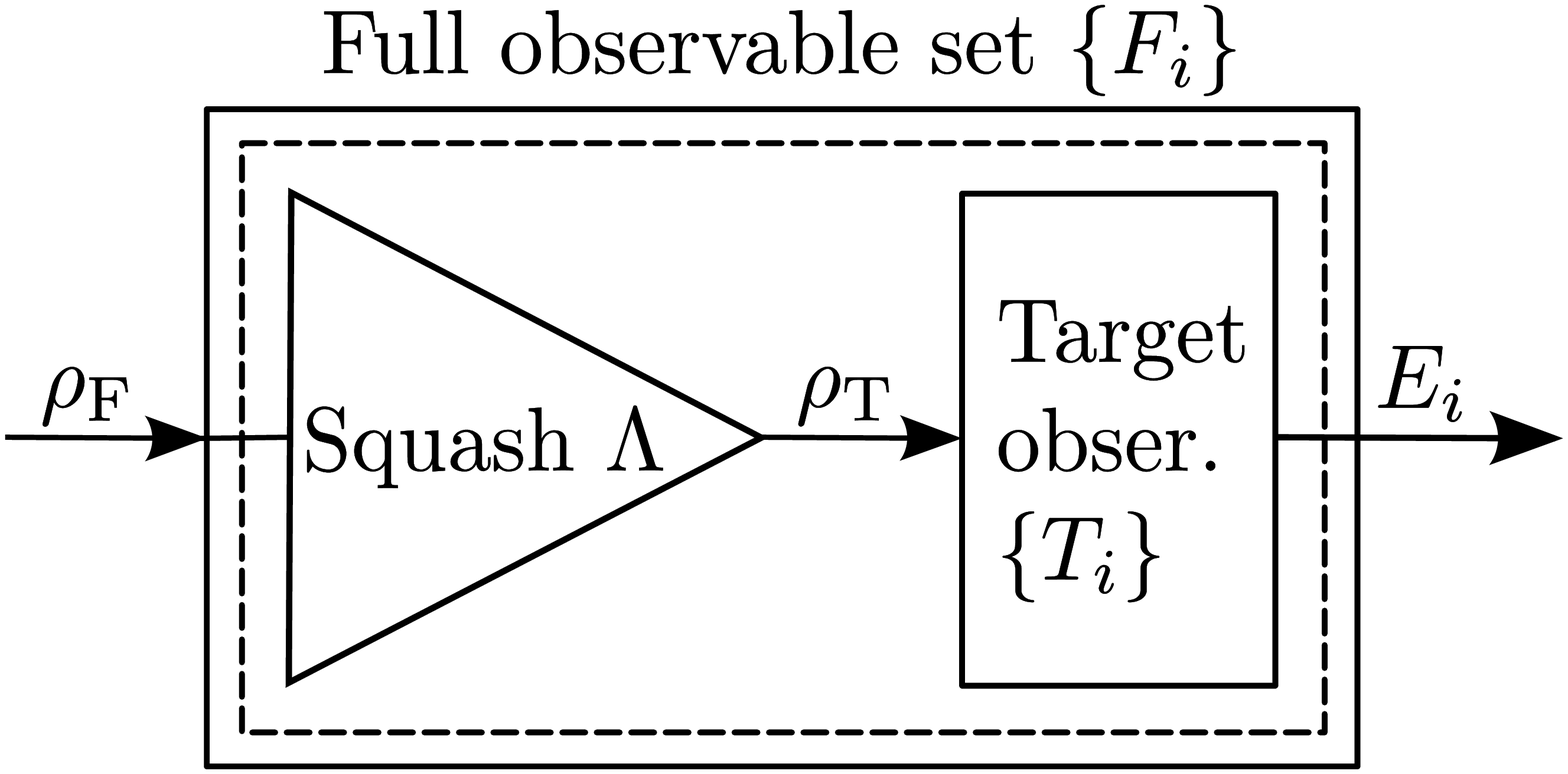}
  \caption{Idea of the positive squashing operation: 
The given full observable set $\{ F_i \}$ can be 
decomposed into a \emph{positive} squashing operator 
$\map$ followed by a particular target observable 
set $\{T_i\}$ such that the same expectation values 
$E_i$ are obtained for all possible input states $\rho_{\rm F}$.}
  \label{fig:squash} 
  \end{center}
\end{figure}

Suppose that both sets of observables are  connected by a positive 
(but not necessarily completely positive) linear map 
$\map: \mathcal{L}(\mathcal{H}_F) \to \mathcal{L}(\mathcal{H}_T)$ 
which satisfies the following: the expectation value of any observable $F_i$ 
with respect to an arbitrary input state $\rho_{\rm F}$ is the same as the 
expectation value of the corresponding target operator $T_i$ with respect to
the output state of the corresponding map 
$\rho_{\rm T}=\Lambda(\rho_{\rm F})$ (see Fig.~\ref{fig:squash}).
That is,
\begin{equation}
\label{eq:map}
\trace(\rho_{\rm F} F_i) = \trace[\map(\rho_{\rm F}) T_i]
\end{equation}
holds for any input state $\rho_{\rm F}$ and for all $i=1,\dots, n$. 
Using the adjoint map 
$\map^\dag: \mathcal{L}(\mathcal{H}_T) \to \mathcal{L}(\mathcal{H}_F)$ 
this condition can be rephrased into
\begin{equation}
\label{eq:adjointmap}
\map^\dag(T_i)=F_i
\end{equation}
for all $i=1,\dots,n$, while the positivity requirement 
transfers also to the adjoint map $\map^\dag$. 

Such a described connection between two different observables 
sets is an extension of
the notion of a squashing operation in the QKD context 
\cite{squash1,squash2} 
which differs from the present definition only by the extra condition of 
being completely positive and trace-preserving. Because of these 
similarities we use the term \emph{positive squashing operation} in 
order to refer to map $\map$ (or its adjoint $\map^\dag$). 
Note that typically we consider the case of a trace-preserving 
map $\map$ (or unital map $\map^\dag$) such that density operators 
are mapped to properly normalized density operator; however this 
requirement is not mandatory. An example of a non-trace preserving, 
but positive map between operator sets is given by the matrix 
of moments \footnote{Though there are different applications of 
the matrix of moments, or equivalently the expectation value 
matrix, only the one from Ref.~\cite{miranowicz} exploits it 
in the same spirit as for the present purpose: Rather than 
trying to reconstruct the matrix of moments of the, \eg, 
partially transposed state \cite{vogel,haeseler}, the 
verification method from Ref.~\cite{miranowicz} applies 
the separability criteria \emph{directly} onto the matrix 
of moments, since it can be considered as an unnormalized 
physical state. Moreover let us note that the matrix of 
moments is the composition of a completely positive map 
followed by the transposition, hence only positive but 
not completely positive.}~\cite{miranowicz}; the only 
difference is that one must be careful with entanglement 
criteria on the target space that explicitly employ the 
normalization of the density operators (\eg, the computable 
cross norm or realignment criterion), but one can 
also deal with this \cite{miranowicz}.

The advantage of such a positive squashing operation 
is that the structure of separable states (from the full 
to the target Hilbert space) remains invariant, and hence 
any successful entanglement verification on the target 
space directly translates to a positive verification 
statement on the full Hilbert space:

\begin{proposition}[Entanglement verification]\label{prop:entveri}
Let us assume that the two sets of local observables on Alice's side, 
labeled as $\{ T^{\rm A}_i \}$ and $\{ F^{\rm A}_i\}$ respectively, 
are connected by a positive (not necessarily completely positive) 
unital linear map $\map_{\rm A}^\dag$ satisfying Eq.~(\ref{eq:adjointmap}), 
and a similar relation holds for Bob's side. If the observed data 
verify entanglement with respect to the target observables 
$T^{\rm A}_i \otimes T^{\rm B}_j$, then this data also proves 
the presence of entanglement for the full operator set 
$F^{\rm A}_i \otimes F^{\rm B}_j$. An analogous statement holds 
for more than two particles.
\end{proposition}

\begin{proof}
For the observed data $E_{ij}$ one has the identity 
$E_{ij}=\trace(\rho_{\rm AB} F^{\rm A}_i \otimes F_j^{\rm B}) 
= \trace[\map_{\rm AB}(\rho_{\rm AB}) T^{\rm A}_i \otimes T_j^{\rm B}]$ 
due to the property of the squashing operation. For any separable 
state on the full bipartite Hilbert space 
$\rho_{\rm AB}^{\rm sep} = \sum_k p_k \rho_{\rm A}^k \otimes  \rho_{\rm B}^k$, 
one obtains 
\begin{equation}
\sigma_{\rm AB}^{\rm sep} :=  \map_{\rm AB}(\rho_{\rm AB}^{\rm sep}) =\sum_k p_k \map_{\rm A}(\rho_{\rm A}^k) \otimes  \map_{\rm B}(\rho_{\rm B}^k),   
\end{equation}
which represents a valid (normalized) separable density operator 
on the bipartite target Hilbert space because of positivity of the 
corresponding (unital) maps, and is compatible with the observed data. Consequently, 
if one proves the incompatibility of the mean values of the $T_i$ with all 
separable states on the target space, the density matrix on the full space 
must be entangled. Note that here one just needs positivity of $\map_{\rm A}$ 
and $\map_{\rm B}$ and not complete positivity. 
\end{proof}

Note that a local squashing operation between operator 
sets does not represent the most general map between bipartite 
observable sets that preserve the structure of separable states; 
however we neglect other options on behalf of the ``locality'' 
of this connection. Furthermore note that since we do not require 
for a completely positive map, it can happen that one obtains an 
unphysical (not positive semidefinite) density matrix on the 
target space; such an operator is then also incompatible 
with a separable state. However this situation can only 
occur for an entangled state on the full bipartite Hilbert 
space, hence the conclusion of the entanglement verification 
process remains unaffected. 

Finally, let us add that the precise state reconstruction 
technique needed for the tomography entanglement test, either 
direct inversion of Born's rule or maximum likelihood 
estimation \cite{hradil97a} (although there are even 
problems associated with them \cite{blume06a}), does not 
conflict with a positive but not completely positive 
squashing operation. If the corresponding operator on 
the target space is positive semidefinite both 
reconstruction techniques deliver the same operator 
(in the limit where one really obtains exact knowledge 
of the expectation value). Because any separable state 
is represented by a valid separable target state this 
excludes the possibility that a separable state is 
mapped to an entangled state by the reconstruction 
process. In the case of an unphysical ``entangled'' 
target operator a direct inversion of Born's rule 
one would directly ``witness'' the entanglement 
\footnote{In this case one should be convinced that the 
actual operator description $T_i^{\rm A} \otimes T_j^{\rm B}$ 
cannot be the {correct} one for the experiment. However via 
the matrix of moments of the partially transposed state \cite{vogel, haeseler} 
one effectively performs such a detection.}. In contrast 
the maximum likelihood method produces the closest positive 
semidefinite operator \cite{blume06a} (with respect to the 
likelihood ``distance''), hence an unphysical, entangled 
target state can be mapped to a separable state via this 
reconstruction technique and thus escapes the tomography 
entanglement test. But this does not bother us here, 
because some entangled states are missed anyway due to the 
simplified operator set.

\section{Criteria for the existence of a positive 
squashing operation}

In this section we investigate which requirements need to be fulfilled by the two different operator sets in order to be connected by a positive squashing operation. There are, of course, different ways how one can tackle this problem: One method, in close analogy to that of Ref.~\cite{squash2}, is to employ the Choi-Jamio{\l}kowski isomorphism \cite{pillis-67,jamiolkowski,choi-82} between linear maps and linear operators. This isomorphism transforms positive maps \textrm{red}{into} entanglement witnesses, or more precisely into linear operators that are positive on product states, while the requirements from Eq.~(\ref{eq:adjointmap}) change into a set of linear equations that constrain the allowed form of the entanglement witness. For an explicit solution to this reformulated problem one first solves these linear equations and afterwards tries to choose the remaining, undetermined parameters of the operator in such a way that it meets the entanglement witness condition.

However, we take a different path that provides a clear interpretation for the existence of such a positive linear map and which is also employed in the later part to prove the positive squashing property for the polarization measurements.

Equation~(\ref{eq:map}) directly allows us to read off a necessary 
condition: it states that for each physical density operator 
$\rho_{\rm F}$ in the full Hilbert space there exists a 
valid density operator $\map(\rho_{\rm F})$ (if $\map$ is 
trace-preserving) in the target space such that both 
operators assign the same expectation values for the 
considered observables. Hence all possible expectation 
values $E_i$ that can in principle be observed on the 
full Hilbert space must remain physical with respect to 
the target observables. As we will see, this condition 
becomes also sufficient if the target operators $T_i$ 
with $i=1,\dots, n$ provide a complete tomographic set. 
Thus, in combination with Prop.~\ref{prop:entveri}, 
we have the following solution for the question posed 
in the introduction: \textit{The tomography entanglement test 
is error-free as long as the full local observables 
on Alice and Bob's side can only produce measurement 
results which are also consistent with the local target, 
or reconstruction observables}. 

For the following proposition we need to define the set of possible 
physical expectation values associated with a given set 
of observables, defined as  
\begin{eqnarray}
\label{eq:setS}
\mathcal{S}_{F} &=& \Big\{ \vec E \in \mathbb{R}^{n} \big| 
\mbox{there is a } \rho \in \density{F} \mbox{ such that } 
\nonumber
\\
&&
\;\;E_i=\trace(\rho F_i)\mbox{ for all } i=1,\dots n \Big\}, 
\end{eqnarray}    
and a similar definition for the operator set on the target 
system $\mathcal{S}_{\rm T}$. Concluding we have the following 
characterization:

\begin{proposition}[Existence] 
\label{prop:necsuff} 
The set of full observables $\{ F_i \}$ and the tomographically
complete set of target observables $\{ T_i \}$ are related 
by a unital squashing operation $\map^\dag$
if and only if it holds that 
$\mathcal{S}_{F}\subseteq \mathcal{S}_{T}$. 
\end{proposition}

\begin{proof}
One direction of the proof is clear: Suppose that there exists 
a positive trace-preserving squashing operation $\map$. For any 
$\vec E \in \mathcal{S}_F$ we must have a density operator $\rho_{\rm F}$ such 
that one obtains 
$E_i=\trace(\rho_{\rm F} F_i)= \trace[ \map(\rho_{\rm F}) T_i]$. 
Because of the properties of the corresponding map we receive a 
valid target density operator $\rho_{\rm T}:=\map(\rho_{\rm F})$ 
which provides the same expectation values $\vec E$, 
hence $\vec E \in \mathcal{S}_T$. This concludes the first 
direction of {$\mathcal{S}_F \subseteq \mathcal{S}_T$}. 

For the other direction, we employ the fact that the set of 
target operators are tomographically complete and the set 
inclusion {$\mathcal{S}_{F}\subseteq \mathcal{S}_{T}$}
to explicitly write down the positive squashing operation. 
First note that for a given set of physical expectation 
values $\vec{E} \in \mathcal{S}_{\rm T}$, the corresponding 
target density operator is uniquely determined by a direct 
inversion of Born's rule, 
$\mathcal{R}_{\rm T}: \vec E \mapsto \rho_{\rm T}(\vec E)$, \ie, 
by a linear reconstruction mechanism that maps the 
expectation values to its explicit density operator. 
Moreover for a given full density operator $\rho_{\rm F}$ 
the corresponding expectation values are already determined, 
which is described by the linear map 
$\mathcal{M}_{\rm F}: \rho_{\rm F} \mapsto \vec E$. 
Combining these two maps according to 
\begin{equation}
\map = \mathcal{R}_{\rm T} \circ \mathcal{M}_{\rm F}
\end{equation}
provides the squashing operation: That is, for a given 
input state $\rho_{\rm F}$ one first computes the 
expectation values $E_i$ of the full operator set 
and then uses these values in the reconstruction 
algorithm (that depends on the target operators) 
to obtain the corresponding target output state. 
The set inclusion guarantees that any valid full 
density operator is mapped to a valid target 
state, hence the described map is already positive. 
Since both maps in the decomposition are linear 
the overall map is linear as well. 
\end{proof}

In a concrete example the proposition of 
course only helps if one obtains knowledge on the sets 
$\mathcal{S}$; although this is by far not a trivial task, 
one can employ approximation techniques for a special set of 
observables or even a hyperplane characterization for the 
exact determination, see Ref.~\cite{morodertmpsu2} for more 
details. 

{If the set of considered observables on the target space is not tomographically complete, then, in order to establish the existence of a positive mapping between the two sets of operators, one can still invoke Proposition \ref{prop:necsuff} with some caution. Indeed, one has to check whether it is possible to extend the two sets by some choice of additional target and full operators, so that the target set is tomographically and the two sets of physical expectation values---which depend on the choice of the extensions---satisfy the condition of Proposition \ref{prop:necsuff}}.

Finally, let us note that one can also characterize a completely positive map via such a set inclusion requirement if one adds an additional reference system $R$ of dimension equal to that of the full space (or of the target space, in the case the dual map) on each side, because complete positivity of $\map$ just means that $\rm{id}_R \otimes \map$ is positive for such a reference $R$. For an actual investigation of such a completely positive map, however, the formalism of  Ref.~\cite{squash2} seems more appropriate to us.

\section{Example: Polarization measurements}

In this section, we apply the developed formalism to a physical 
relevant measurement setup. We draw our attention to polarization 
measurements onto a two-mode system by using only threshold detectors,
 \ie, such detectors cannot resolve the number of photons. More 
precisely, as shown in Fig.~\ref{fig:detectionsetup}, the incoming 
light field is separated according to a chosen polarization basis 
$\beta \in \{x,y,z\}$ via a polarizing beam splitter, followed 
by a photon number measurement on each of those modes by a simple 
threshold detector. Hence in total four different outcomes can be 
distinguished: no click at all, only one of the detector clicks or 
both of them trigger a signal and produce a double click.  Because 
of its great simplicity this measurement device appears quite 
frequently in quantum optical experiments which employ the 
polarization degree of freedom (for an overview see 
Ref.~\cite{pan-2008}). It turns out that this measurement 
device provides, if measured along all three different 
basis settings, a non-trivial example for the difference between
a positive and a completely positive squasher. This means
that the corresponding map $\map$ can only chosen to be 
positive but not completely positive. 

\begin{figure}
  \begin{center}
  \includegraphics[scale=1.2]{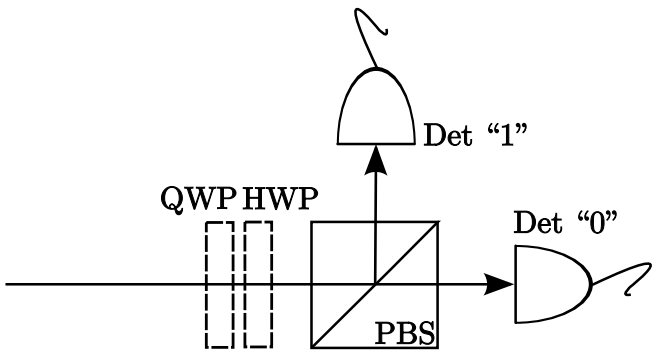}
  \caption{Schematic setup of the considered polarization measurement:
 Via quarter (QWP) and half wave plate (HWP) one can effectively 
adjust the polarization basis $\beta$ of the corresponding 
polarizing beam splitter (PBS) according to the basis 
$\{ \pm 45^{\circ}, \circlearrowright / \circlearrowleft, H/V \}$ 
that we label as $\{x,y,z\} $.}
  \label{fig:detectionsetup} 
  \end{center}
\end{figure}

Next let us specify which observable sets should be connected by 
the squashing operation; see also Ref.~\cite{squash2}. For each 
chosen polarization basis $\beta$ the different mode measurements 
are denoted as $M_{i,\beta}$ with the label $i \in \{\rm {vac}, 0,1,d\}$ 
for all classical outcome possibilities: 
vacuum, click in ``0'', click in ``1'', and a double-click. 
The target measurements are chosen such that they justify a 
single photon description: each click event is interpreted 
as a single photon state, hence as target measurements 
one selects the same measurement device, only that it just 
acts on the single photon subspace and the vacuum component. 
In order to achieve the squashing property, the double 
click events must be taken into account, since such events 
are clearly incompatible in a (perfect) single photon interpretation, 
but they nevertheless contribute to the normalization. One can 
incorporate this effect by a particular post-processing 
scheme that represents a sort of penalty for double click 
events. A common scheme, originally introduced for the 
QKD context in Ref.~\cite{nor_99,nor_00}, consists of 
randomly assigning each double click event to one of 
the single click outcomes. This particular set of 
processed measurement operators becomes the exact 
set of full operators $\{ F_{i,\beta}\}$ with 
$i \in \{ \rm {vac},0,1\}$ and $\beta \in \{x,y,z\}$ 
and with the relation $F_{i,\beta}=M_{i,\beta}+1/2 M_{d,\beta}$ 
with $i=0,1$ for all $\beta$. 

Let us start with a perfect single-polarization mode description 
of the full operators; imperfections like finite efficiency or 
dark counts are considered later (see also Ref.~\cite{nor_99}). 
The ``no click'' outcome is independent of the 
chosen polarization basis and becomes 
$F_{{\rm vac},\beta}=\ket{0,0}\bra{0,0}$. 
All other observables are block-diagonal with respect to the photon 
number subspace, \ie, $F_{i,\beta}=\sum_{n=1}^\infty F_{i,\beta}^{n}$ 
and for a fixed number of photons we have
\begin{equation}
\label{eq:F_ibeta}
F_{i,\beta}^{n} = \frac{1}{2} \left[ \mathbbm{1}_n \! + \!(-1)^i \left( \ket{n,0}_\beta\bra{n,0} - \ket{0,n}_\beta\bra{0,n} \right)\right]\!,\!\!
\end{equation}
with $i=0,1$. Here $\ket{k,l}_\beta$ denotes the corresponding 
two-mode Fock state in the chosen polarization basis $\beta$ 
(\eg, for $n=3$ the state $\ket{2,1}_z = \ket{2H,1V}$ describes 
a system with two horizontally and one vertically polarized photon) 
and $\mathbbm{1}_n$ represents the identity operator onto the 
$n$-photon subspace, which appears because of the chosen 
post-processing scheme. This perfect single-polarization mode 
description is also employed for the target operators, however 
only acting on the vacuum $T_{{\rm vac},\beta}=\ket{0,0}\bra{0,0}$ 
or on the single photon subspace $T_{i,\beta}=F_{i,\beta}^{1}$ 
with $i=0,1$. 

Let us further comment on these observable sets: Note that if one 
selects the following standard basis for the single photon subspace 
$\ket{1,0}_z=\ket{0}$ and $\ket{0,1}_z=\ket{1}$, then each difference 
of the single click outcomes equals to a familiar Pauli operator, 
\ie, $\sigma_\beta=T_{0,\beta}^{1}-T_{1,\beta}^{1}$ for all $\beta$. 
Hence each of the single click operators $T_{i,\beta}$ with $i=0,1$ 
corresponds to a projection onto one of the two different eigenstates 
of the related Pauli operator $\sigma_\beta$. Furthermore the 
corresponding difference between the full observables 
$F_\beta = F_{0,\beta}-F_{1,\beta}$ is again block-diagonal 
and each $n$-photon part is given by 
\begin{equation}
\label{eq: Sn}
F^{n}_\beta = F_{0,\beta}^n-F_{1,\beta}^n = \ket{n,0}_\beta\bra{n,0}- \ket{0,n}_\beta\bra{0,n},
\end{equation}
according to Eq.~(\ref{eq:F_ibeta}). Note that these observables are 
also accessible with a different polarization measurement that only 
uses a single threshold detector \footnote{The measurement setup is 
similar to the one from Fig.~\ref{fig:detectionsetup}, however one 
only measures {behind} one of the outputs of the polarizing beam splitter. 
It is {straightforward to check} that the operators $F_\beta$ can be obtained by using 
the difference on the two outputs (or alternatively with appropriate adjusted {wave plates}). 
However in order to obtain the normalization, respectively the identity $\mathbbm{1}$, one 
has to measure the overall input via a threshold detector, \ie, 
with no polarizing beam splitter. It is not, as typically employed, 
given by the sum of both clicks events on both different outcomes.}
 and which has alternatively been employed for polarization 
experiments, cf. Ref.~\cite{james01a}. 

The following theorem proves the positive squashing 
property between the two given sets of observables; 
however it also applies to the other measurement 
description of Ref.~\cite{james01a}.

\begin{theorem} \label{thm:polarization} There exists a positive, 
but not completely positive unital squashing operation $\map^\dag$ for 
the operator sets $\{T_{i,\beta}\}$ and $\{F_{i,\beta}\}$, 
\ie, $\map^\dag(T_{i,\beta})=F_{i,\beta}$. Therefore, the 
interpretation of the $\{F_{i,\beta}\}$ as 
single photon measurements $\{T_{i,\beta}\}$ does not invalidate 
the entanglement verification scheme.
\end{theorem}

\begin{proof}
First let us point out that the existence of 
a \textit{completely} positive squashing operation 
has already been ruled out in Ref.~\cite{squash2}. 

In order to prove the existence of a positive squashing operation 
we only need to focus on the ``click'' events, since the vacuum part 
can be directly removed by a projection discriminating between 
the vacuum and all other Fock-states. Note that it is sufficient 
to prove the squashing operation for a complete set of linear independent 
target operators only, because other linear dependencies are 
implicitly present in the linear map. In short, it is equivalent 
to prove a unital squashing operation $\map^\dag(\sigma_\beta)=F_\beta$ 
for all $\beta \in \{x,y,z\}$, where $F_\beta$ is the described 
difference between the click outcomes of the full observables. 

Since we only concentrate on the single photon subspace we are 
equipped with a full tomographic set and hence can readily apply 
Prop.~\ref{prop:necsuff}, such that it remains to prove 
$\mathcal{S}_F \subset \mathcal{S}_T$. Since each full 
observable is photon number diagonal one obtains that  
$\mathcal{S}_F$ is given by the convex hull of all $n$-photon 
sets $\mathcal{S}_F^{n}$, \ie, the set of physical expectation 
values on an $n$-photon state. Hence we need to verify that 
each $n$-photon state can only produce expectation values 
which are also compatible with a single photon state, 
\ie, $\mathcal{S}_F^n \subset \mathcal{S}_F^{1}=\mathcal{S}_T$ 
for all $n \geq 1$. The set $\mathcal{S}_F^{1}$ directly equals 
to the familiar Bloch sphere. Hence we prove existence of a 
positive squashing operation if we can show that 
\begin{equation}
\label{eq:bloch}
\sum_{\beta \in \{x,y,z\}} [\trace(\rho F_\beta^n )]^2  \leq 1
\end{equation}
holds for all $n$ photon density operators $\rho$, and for all 
photon numbers $n\geq 1$.

In order to simplify the analysis in the following, each operator 
$F_\beta^n$ can be regarded as an operator acting onto an 
$n$-qubit space. Indeed, the $n$-photon Hilbert space 
$\mathcal{H}_F^n=\mathbbm{C}^{n+1}$ is isomorphic to the 
symmetric subspace $\text{Sym}(\mathcal{H}_n)$ of an 
$n$-qubit system  $\mathcal{H}_n=(\mathbbm{C}^2)^{\otimes n}$. 
Using the given standard basis definition one obtains 
for example  
\begin{equation}
  \label{eq: Sz}
  F^{n}_z=\ket{0}\bra{0}^{\otimes n}- \ket{1}\bra{1}^{\otimes n},
\end{equation}
while for any other operator $F_\beta^{n}$ one just replaces 
the states $\ket{0},\ket{1}$ with the eigenvectors of the 
corresponding Pauli matrix $\sigma_\beta$. 

Expanding these operators in a multi-qubit basis delivers 
\begin{eqnarray}
\nonumber
F^{n}_\beta&=& \left(\frac{\mathbbm{1}+\sigma_\beta}{2}\right)^{\otimes n}- \left(\frac{\mathbbm{1}-\sigma_\beta}{2}\right)^{\otimes n}\\
\label{eq:neu} &=& \frac{1}{2^{n-1}}\sum_{j \:\text{odd}} \sum_\pi \pi\left(\sigma_\beta^{\otimes j} \otimes \mathbbm{1}^{\otimes (n-j)}\right)
\end{eqnarray}
in which $\sum_\pi$ denotes the sum over all possible 
permutations $\pi(\cdot)$ of the subsystems that yield 
different terms. 

Next, we exploit the result from 
Ref.~\cite{toth05} that for odd $j$ every quantum state $\rho$, 
hence also each state on the symmetric space, 
satisfies 
\begin{equation}
\sum_{\beta = x,y,z} \braket{\pi(\sigma_\beta^{\otimes j})}_\rho^2 \leq 1, 
\end{equation}
with the abbreviation
\begin{equation}
\braket{\pi(\sigma_\beta^{\otimes j})}_\rho=\trace\left[\rho \:\pi\!\left(\sigma_z^{\otimes j} \otimes \mathbbm{1}^{\otimes (n-j)}\right)\right].
\end{equation}
This inequality is based on the property that the observables 
$\pi(\sigma_\beta^{\otimes j} \otimes \mathbbm{1}^{\otimes (n-j)})$ 
with $\beta \in \{x,y,z\}$ have all eigenvalues equal to $\pm 1$ 
and anti-commute pairwise \footnote{For completeness, let 
us recall the proof: Let $M_i$ be anti-commuting observables 
(\ie, $M_i M_j + M_j M_i = 0$ for all $i \not = j$) with 
$M_i^2=\mathbbm{1}$ for all $i$ and let $\lambda_i$ be 
real coefficients with $\sum_i \lambda_i^2 = 1.$ 
Then $(\sum_i \lambda_i M_i)^2 = \mathbbm{1}.$ 
Therefore, 
$(\sum_i \lambda_i \mean{M_i})^2 \leq \mean{(\sum_i \lambda_i M_i)^2} = 1,$ 
hence $\sum_i \lambda_i \mean{M_i} \leq 1,$ and, since 
the $\lambda_i$ are arbitrary, $\sum_i\mean{M_i}^2 \leq 1.$}. 
Note that this identity holds for all occurring $j$ and for 
all possible permutations $\pi$. Consequently one obtains 
\begin{eqnarray}
&& \sum_{\beta} \; [\trace(\rho F_\beta^n )]^2  \\ \nonumber &=& \!\frac{1}{2^{2n-2}}\!\sum^n_{j,j^\prime \text{ odd}} \sum_{\pi,\pi^\prime} \left[ \sum_\beta \braket{\pi(\sigma_\beta^{\otimes j})}_{\rho} \braket{\pi^\prime(\sigma_\beta^{\otimes j^\prime})}_{\rho} \!\right] \leq 1,
\end{eqnarray}
where the inequality (together with the Cauchy-Schwarz 
inequality) was used to upper bound each term in the 
squared bracket by $1$. For the final result one 
needs to count the numbers of distinct 
permutations $\pi$, which is given by a 
corresponding binomial coefficient. 
\end{proof}

How can one use this result in the tomography entanglement 
test of a PDC source? First each party measures along all 
three different polarization axes. Next one either actively 
post-processes the double click events or just passively 
computes the corresponding rates and/or probabilities of 
the full operators. Afterwards both parties can safely use 
the single photon assumption, or more precisely, the set 
of perfect single photon target operators $\{ T_{i,\beta} \}$ 
to compute the corresponding two-qubit state $\rho_{\rm AB}$ 
(single photon subspace on each side) via their preferred 
reconstruction technique. In case that this reconstructed 
state is entangled one can be assured that the observed 
data still verify entanglement if both parties believe 
in the more realistic measurement description $\{ F_{i,\beta} \}$.

Next let us focus on the imperfections of the photo-detectors. 
Real photo-detectors register only some portion of the 
incoming photons, a significant part is not detected. 
If both detectors in the setup of Fig.~\ref{fig:detectionsetup} 
have the same inefficiency, this inefficiency can be modeled 
by an additional beam splitter in front of the 
perfect measurement device \cite{yurke}, hence if one combines 
the beam splitter map (completely positive) with the already 
proven positive squashing map from the perfect case then one 
directly extends the validity of the positive squashing property 
to an inefficient measurement description. The same idea applies 
to dark counts, which can be modeled as a particular post-processing 
scheme on the classical outcomes \cite{nor_99}, and to misalignment 
errors, that are described by a fixed depolarizing channel onto 
each photon separately. Even the extension to a multi-mode 
description is possible if one employs the model from Ref.~\cite{squash3}. 

Concerning real experiments, one should note that double-clicks in a spatial 
mode can arise from different physical mechanisms. First, it can happen 
that due to the higher orders in the PDC process more than the desired number of photons 
are generated and injected into the setup. Second, dark counts may lead to double click 
events. Finally, double-click events may occur in special setups for the generation of 
certain multipartite states, this case is, however, not important for our 
discussion
\footnote{In some setups, double click events arise from the statistical 
nature of the state preparation: For instance, in Ref.~\cite{lu} entangled multi-photon 
states are generated by producing several entangled photon pairs first, and then letting 
them interact via beam splitters. The desired state is  only produced if all the photons 
are distributed uniformly over all the spatial modes, that is, if each mode contains one 
photon. Due to the statistical properties of the beam splitters, this is not always the 
case, and often one of the spatial modes contains more than just one photon (and a different
mode contains no photon), so that the double click rate at this outcome side 
drastically increases.
However, neglecting these double-clicks is justified: Since in this 
case some spatial mode does not contain any photon, disregarding 
these events is equivalent to projecting the total multi-photon 
state onto the space where each mode contains at least one photon. 
Since this is a local projection, it cannot produce fake 
entanglement.
}.

Then, it is worth mentioning that the post processing used in 
the above scheme is usually not applied in real experiments: 
double click events are typically just thrown away.
In practice, however, the amount of these undesired events is quite 
small: For instance, in the {4-photon} experiment of Ref.~\cite{wieczorek08a} 
the number the coincidences where a double click occurs in one 
mode while in the other three modes there is one click, is around 
0.77 \% of the (desired) events, where in each mode there is exactly 
one click \footnote{Witlef Wieczorek, private communication.}. It should 
be noted, however, that in experiments with more and more photons, 
these rates can be higher \cite{laskowski-2009}, so that the penalty effect of the 
post-processing scheme becomes higher.

Additionally we  comment on two points: As one can prove, 
the corresponding squashing map is \emph{completely positive} on the 
single and two photon subspace \cite{squash2}. Hence one only 
observes a violation of positivity of the corresponding target 
density operator if the local multi-photon contributions  
are very large in comparison to the single and double 
photon part (and even then only for very particular 
entangled states); {consequently, to observe such a non-positive target operator in a real PDC experiment is very unlikely }.

As a last point we should make it clear that one cannot always apply 
Theorem~\ref{thm:polarization}. Especially in multipartite 
experiments, it happens that one does not even want to obtain 
full tomography onto the multipartite target space 
but instead tries to measure an entanglement witness with 
the least number of different global measurement settings. 
This may require more than three different settings on each 
photon. For instance, in the six-photon experiment of 
{Ref.~\cite{lu}} an entanglement witness was measured 
which required seven measurements settings of the type 
$M_i \otimes M_i \otimes ... \otimes M_i$, which is 
a significant advantage compared with the $3^6=729$ 
settings required for state tomography.
However,  on each photon seven polarization measurements 
have been made and the target observables are tomographically 
overcomplete. In such cases this theorem does not 
apply, because the linear dependencies imposed by the 
target operators are not satisfied by the 
full observable set, cf. Eq.~(\ref{eq:adjointmap}), 
hence the local squashing operation does not exist{---in fact the map cannot even be linear.} 
Here one might attempt to proceed with a global, separable squashing operation, cf.~comment after 
Prop.~\ref{prop:entveri}.

\section{Positive squashing and entanglement quantification}

In this section we argue that a local squashing operation, even 
if it is not completely positive, can in principle not only
provide qualitative indications about the presence of entanglement, 
as was proved in Proposition~\ref{prop:entveri},
but also \emph{quantitative} ones.

Let us start by recalling the notion of entanglement measure. An entanglement measure is a function from
density operators to (positive) real numbers, that captures quantitatively some property of entangled states.
There are many entanglement measures in the literature~\cite{plenio2007introduction}; some of them have an operational character,
while some others focus on structural properties of entangled states, for example, the fact that, by definition, they do not admit a separable decomposition. Even if some entanglement measures do not have a direct operational interpretation, they are nonetheless useful because they may provide upper and lower bounds to operational measures or other interesting quantitative properties of entanglement. Furthermore, any entanglement measure can be considered as a benchmark for the quality of an experiment designed to create ``highly entangled'' states and {to} display a good global control on more than one subsystem at a time. This is due to two facts. The first is that, although different entanglement measures do not correspond to the same ordering of states from ``unentangled'' to ``maximally entangled'', there is typically a correlation: a state that is highly entangled with respect to one measure is, in most cases of interest, highly entangled with respect to another one. The second fact is that in an axiomatic approach to entanglement measures, it is typically asked that entanglement, as quantified by some entanglement measure, does not increase under the restricted class of Local Operations and Classical Operations (LOCC). Indeed, entanglement cannot be created by LOCC operations alone, and it is natural to require that {any entanglement measure} does not increase under this set of operations. In this way, on one hand, entanglement is elevated to a resource that by LOCC can only be manipulated and not augmented, and on the other hand, entanglement measures satisfying such an LOCC monotonicity are a fair benchmark for the ability of the experimenters to jointly manipulate many subsystems. 

Let us be more precise about the LOCC monotonicity of entanglement measures, 
focusing on the bipartite case. We say that $E$ is an LOCC monotone if
\begin{equation}
E(\rho_{\rm AB})\geq E(\Lambda_{\rm LOCC}[\rho_{\rm AB}]),
\end{equation}
where $\Lambda_{\rm LOCC}$ is {an LOCC} transformation.
In particular, local completely positive trace-preserving (CPTP)
maps belong to the class of LOCC operations, so that $E$ is 
monotone with respect to CPTP local operations:
\beq
\label{eq:CPlocalmonotone}
E(\rho_{\rm AB})\geq E((\Lambda_{\rm A}\otimes \Lambda_{\rm B})[\rho_{\rm AB}]).
\eeq
Thus, if the squashing is realized by local CPTP maps, the entanglement of the reconstructed squashed state $(\Lambda_{\rm A}\otimes \Lambda_{\rm B})[\rho_{\rm AB}]$ is a lower bound for the entanglement of the physical state actually prepared. The point here is that one can generalize Eq.~\eqref{eq:CPlocalmonotone} to the case of positive but not completely positive maps, at least for the entanglement measure called negativity~\cite{ZyczkowskiHSL98,VidalW02}, which is one of the few entanglement measures that can be easily computed.

The negativity of a bipartite state $\rho_{\rm AB}$ is defined as
\beq
\label{eq:negativity}
N(\rho)=\frac{\|\rho_{\rm AB}^\Gamma\|_1-1}{2},
\eeq
where $\rho_{\rm AB}^\Gamma=(T\otimes\idmap)[\rho_{\rm AB}]$ {denotes the partial transpose of the original density operator}. {Here, $T$ is the transposition, which is a positive but not completely positive trace-preserving map, while ``$\idmap$'' stands for the identity map, and $\
\|X\|_1=\Tr(\sqrt{X^\dagger X})$ is the trace norm on operators}. The value of the negativity is independent of the party we choose to apply transposition to, and quantifies the degree of violation of the partial transposition separability criterion~\cite{peres,entanglement_witness}. Indeed, it corresponds to the sum of the absolute values of the negative eigenvalues of $\rho_{\rm AB}^\Gamma$.

In the Appendix we prove the following inequality
\beq
\label{eq:Pnegativity}
\begin{aligned}
N(\rho_{\rm AB}) \geq \frac{1}{\|\tilde{\Lambda}_{\rm A}\otimes\Lambda_{\rm B}\|^H_1}&\bigg(N((\Lambda_{\rm A}\otimes\Lambda_{\rm B})[\rho_{\rm AB}])\\
&-\frac{\|\tilde{\Lambda}_{\rm A}\otimes\Lambda_{\rm B}\|^H_1-1}{2}\bigg),
\end{aligned}
\eeq
with $\tilde{\Lambda}_{\rm A}=T\circ \Lambda_{\rm A}\circ T$ {being} (completely) positive if and only if $\Lambda_{\rm A}$ is (completely) positive, and {with a norm on linear maps defined as $\|\Omega\|^H_1\equiv\max_{\ket{\psi}:\langle\psi | \psi \rangle=1}\|\Omega[\proj{\psi}]\|_1$ (cf. Ref.~\cite{watrous2005norms}).}
We stress that $\|\tilde{\Lambda}_{\rm A}\otimes\Lambda_{\rm B}\|^H_1$ is a measure of the joint violation of complete positivity of $\tilde{\Lambda}_{\rm A}$  and $\Lambda_{\rm B}$. Indeed, {if both maps $\tilde{\Lambda}_{\rm A}$ and $\Lambda_{\rm B}$ are trace-preserving and completely positive then one obtains $\|\tilde{\Lambda}_{\rm A}\otimes\Lambda_{\rm B}\|^H_1=1$ and one recovers the inequality given by Eq.~\eqref{eq:CPlocalmonotone}}.

Let us remark that $N((\Lambda_{\rm A}\otimes\Lambda_{\rm B})[\rho_{\rm AB}])$ is the negativity, as defined by Eq.~\eqref{eq:negativity}, of the Hermitian operator $(\Lambda_{\rm A}\otimes\Lambda_{\rm B})[\rho_{\rm AB}]$. If {the local squashing operations} $\Lambda_{\rm A}$ and $\Lambda_{\rm B}$ are not completely positive, {then the latter need not be a physical state because of negative eigenvalues even before the partial transposition. The correcting terms in Eq.~\eqref{eq:Pnegativity}, with respect to Eq.~\eqref{eq:CPlocalmonotone}, in particular the presence of $\|\tilde{\Lambda}_{\rm A}\otimes\Lambda_{\rm B}\|^H_1$, take care of this possibility, making inequality~\eqref{eq:Pnegativity} hold.}

{As shown in the Appendix, a different and possibly weaker} lower bound on the negativity is given by
\beq
\label{eq:Pnegativitydiamond}
\begin{aligned}
N(\rho_{\rm AB})
\geq 
\frac{1}{\|\Lambda_{\rm A}\|_\diamond\|\Lambda_{\rm B}\|_\diamond}&\bigg(N((\Lambda_{\rm A}\otimes\Lambda_{\rm B})[\rho_{\rm AB}])\\
&-\frac{\|\Lambda_{\rm A}\|_\diamond\|\Lambda_{\rm B}\|_\diamond-1}{2}\bigg).
\end{aligned}
\eeq
{Here $\|\Omega\|_\diamond\equiv\|\Omega\otimes\idmap\|_1$ is the diamond norm~\cite{diamondnorm},} {with the identity map that can be considered as acting on the same input space as $\Omega$}, and $\|\Omega\|_1\equiv\sup_{\|X\|_1\leq 1}\|\Omega[X]\|_1$ the trace norm for maps.

We further remark that, in the case of a positive but not completely positive squashing {operation}, it might be possible to obtain lower bounds for the entanglement of the original state also for other entanglement measures. Although we are unable to provide further explicit examples at this time, we observe that this might be true for the relative entropy of entanglement~\cite{relent1,relent2}. The latter is defined for a state $\rho_{\rm AB}$ as
\begin{equation}
E_R(\rho_{\rm AB})=\min_{\sigma_{\rm AB}^{\rm sep}}S(\rho_{\rm AB}\|\sigma_{\rm AB}),
\end{equation}
where the minimum runs over all separable states and $S(\rho_{\rm AB}\|\sigma_{\rm AB})=\Tr[\rho_{\rm AB}(\log_2\rho_{\rm AB}-\log_2\sigma_{\rm AB})]$ is the relative entropy. Monotonicity of {this measure} under CPTP LOCC operations can be easily checked as follows:
\beq
\begin{aligned}
E_R(\rho_{\rm AB})&=\min_{\sigma_{\rm AB}^{\rm sep}}S(\rho_{\rm AB}\|\sigma_{\rm AB}^{\rm sep})\\
	&\geq \min_{\sigma_{\rm AB}^{\rm sep}}S(\Lambda_{\rm LOCC}[\rho_{\rm AB}]\|\Lambda_{\rm LOCC}[\sigma_{\rm AB}^{\rm sep}])\\
	&\geq \min_{\tau_{\rm AB}^{\rm sep}}S(\Lambda_{\rm LOCC}[\rho_{\rm AB}]\|\tau_{\rm AB})\\
	&=E_R(\Lambda_{\rm LOCC}[\rho_{\rm AB}]).
\end{aligned}
\eeq
In the first inequality one uses monotonicity of the relative entropy under CPTP maps; {for the second inequality one employs} the fact that a CPTP LOCC map transforms a separable state into another separable state. Now, a local map $\Lambda_{\rm A}\otimes\Lambda_{\rm B}$ also transforms a separable state into a separable state as soon as {the maps} {$\Lambda_{\rm A}$ and $\Lambda_{\rm A}$ are} positive and trace-preserving---this was the key fact used in Proposition \ref{prop:entveri}. If monotonicity of the relative entropy holds under some local map $\Lambda_{\rm A}\otimes\Lambda_{\rm B}$, even if $\Lambda_{\rm A}$ and $\Lambda_{\rm B}$ are just positive but not completely positive, then the inequality $E_R(\rho_{\rm AB})\geq E_R((\Lambda_{\rm A}\otimes\Lambda_{\rm B})[\rho_{\rm AB}])$ still {remains true}. This possibility is left open {for example by the fact that the requirement often used to prove monotonicity of the relative entropy is not complete positivity, but the weaker request of 2-positivity~\cite{petz2003monotonicity} (together with a trace preservation condition). A given map $\Omega$ is 2-positive if} 
\begin{equation}
\left( \begin{array}{cc} A & B\\ C & D  \end{array} \right)\geq 0
\Rightarrow
\left( \begin{array}{cc}
\Omega[A] & \Omega[B] \\ \Omega[C] & \Omega[D] 
\end{array} \right)\geq 0,
\end{equation}
 where $A$, $B$, $C$ and $D$ are matrices themselves. {Hence, if both maps} $\Lambda_{\rm A}$ and $\Lambda_{\rm B}$ are positive and trace-preserving, and the combined local map $\Lambda_{\rm A}\otimes\Lambda_{\rm B}$ is 2-positive, then the inequality $E_R(\rho_{\rm AB})\geq E_R((\Lambda_{\rm A}\otimes\Lambda_{\rm B})[\rho_{\rm AB}])$ still holds.
 
In conclusion, a positive squashing operation does not only provide qualitative statements about entanglement, but potentially also quantitative ones. Open problems regard the application of the derived bounds on the negativity to specific cases, and the analysis of other entanglement measures. Let us remark that in case of the negativity a detailed analysis of lower bounds on the entanglement essentially deals  with the issue of separating the negativity due to the application of the local squashing maps from the negativity due to partial {transposition}. As the bounds are conservative, only states that are {sufficiently} entangled may result in a non-trivial lower bound. Indeed, it is clear that if $N((\Lambda_{\rm A}\otimes\Lambda_{\rm B})[\rho_{\rm AB}])=0$---this is the case for a separable $\rho_{\rm AB}$ and positive $\Lambda_{\rm A}$ and $\Lambda_{\rm B}$---and $\|\tilde{\Lambda}_{\rm A}\otimes\Lambda_{\rm B}\|^H_1>1$ or $\|\Lambda_{\rm A}\|_\diamond\|\Lambda_{\rm B}\|_\diamond>1$, respectively, then the right-hand sides of Eq.~\eqref{eq:Pnegativity} and Eq.~\eqref{eq:Pnegativitydiamond} are actually negative. It is {worth remarking that in the derivation of the bounds for the negativity we have not made use of the positivity of the squashing operations}. This indicates that if one considers local squashing operations with the aim of entanglement verification and quantification, then one may hope to be able to further relax the requirements on the corresponding maps, not only going beyond complete positivity, but also beyond positivity if adequate care is taken.


\section{Conclusions} 

Entanglement verification typically assumes that one knows the underlying measurement operators so that each classical outcome gets an accurate quantum mechanical interpretation. We have addressed the question under what conditions an affirmative entanglement statement remains valid if a simplified description of the measurement apparatus is used. This situation can occur if the actual measurement observables are different from the ones used in the verification analysis, simply because of imperfections {or wrong calibration}. However it 
can even happen on purpose: Indeed one can try, despite being aware of certain differences, to explain the data 
via an oversimplified model, \eg, a very low-dimensional 
description, that eases then the task of applying an 
entanglement criteria. Such a case occurs for example for 
an active polarization measurement with threshold detectors 
to analyze the entanglement from a PDC source. Here one may choose 
 a single photon description only, although 
one knows that certain multi-photon states can also trigger 
events that are indistinguishable from a single photon case, 
because then one easily obtains ``tomography'' by using three 
different measurement settings and directly checks for 
entanglement on the reconstructed two-qubit state. 

Summarizing, a positive entanglement statement remains valid
if the two operator sets can be related by a positive (not 
necessarily completely positive) map. In case that the 
reconstruction operators provide complete tomography such 
a positive maps exists if and only if all measurement results 
from the refined, actual measurement devices are compatible 
with the assumed measurement description of the device.
We have shown that the aforementioned polarization 
measurement, measured along all three different polarization 
axes, constitutes a physical relevant example of such a connection 
that is positive but not completely positive. This result shows 
that most of the performed tomography entanglement tests for a 
PDC source are indeed error-free if one incorporates a penalty for double clicks 
during the state reconstruction. This verifies entanglement for a more realistic model, with 
imperfections and multi-photon contributions, of the measurement used.
Finally, we argued that  it might be possible to obtain not only a positive qualitative statement
about the presence of entanglement, but also a quantitative one, even in cases where the squashing
map is not completely positive and standard results about monotonicity of entanglement measures
can not directly be used.

\section{Acknowledgments}
The authors wish to thank M. Junge, R. Kaltenbaek, C. Roos and W. Wieczorek 
for very useful discussions. This work was funded by the European 
Union (OLAQUI, QAP, QICS, SCALA, SECOQC), the NSERC Innovation 
Platform Quantum Works, the Ontario Centres of Excellence and 
the NSERC Discovery grant and the FWF (START prize).


\appendix

\section*{Appendix}

For a Hermitian operator $\rho_{\rm AB}$ normalized to satisfy $\Tr(\rho_{\rm AB})=1$, we define negativity as
\begin{equation}
N(\rho_{\rm AB})=\frac{\|\rho_{\rm AB}^\Gamma\|_1-1}{2}
\end{equation}
where $\rho_{\rm AB}^\Gamma=(T\otimes\idmap)[\rho_{\rm AB}]$, and $T$ is the transposition. Negativity corresponds to the sum of the negative eigenvalues of $\rho_{\rm AB}^\Gamma$.

Any Hermiticity preserving map $\Lambda$ can be written as $\Lambda[X]=\sum_ic_iK_i X K^\dagger_i$, $c_i\in \mR$. Then $T\circ \Lambda=\tilde{\Lambda}\circ T$, with $\tilde{\Lambda}:X\mapsto\sum_ic_iK^*_i X K^T_i$, i.e., $\tilde{\Lambda}=T\circ\Lambda \circ T$. If $\Lambda$ is (completely) positive, that is  if $c_i\geq 0$ for all $i$, then $\tilde\Lambda$ is (completely) positive. It also holds that $\Lambda$ is trace-preserving if and only if $\tilde{\Lambda}$ is trace-preserving.

For any map $\Omega$ we define {the norm} $\|\Omega\|^H_1\equiv\max_{\ket{\psi}:\langle\psi|\psi\rangle=1}\|\Omega[\proj{\psi}]\|_1$~\cite{watrous2005norms}. Moreover, we observe that the trace norm of a Hermitian operator $X$ can be expressed as $\|X\|_1=\max_{-\openone\leq M\leq \openone} \Tr(MX)$.
For any $-\openone\leq M\leq \openone$,
\beq
\begin{aligned}
|\bra{\psi}\Omega^\dagger[M]\ket{\psi}|&=|\Tr(\Omega^\dagger[M]\proj{\psi})|\\
	&=|\Tr(M\Omega[\proj{\psi}])|\\
	&\leq\|\Omega\|^H_1.
\end{aligned}
\eeq
Therefore, if $-\openone\leq M\leq \openone$, then $-\openone\leq \frac{\Omega^\dagger[M]}{\|\Omega\|^H_1}\leq \openone$.

Thus, assuming that $\Lambda_{\rm A}$ and $\Lambda_{\rm B}$ are trace-preserving---so that $\Tr((\Lambda_{\rm A}\otimes\Lambda_{\rm B})[\rho_{\rm AB}])=1$:
\begin{multline}
N((\Lambda_{\rm A}\otimes\Lambda_{\rm B})[\rho_{\rm AB}])\\
\begin{aligned}
&=\frac{\|((T\circ\Lambda_{\rm A})\otimes\Lambda_{\rm B})[\rho_{\rm AB}]\|_1-1}{2}\\
	&=\frac{\|(\tilde{\Lambda}_A\otimes\Lambda_{\rm B})[\rho_{\rm AB}^\Gamma]\|_1-1}{2}\\
	&=\frac{1}{2}\bigg\{\max_{-\openone\leq M\leq \openone} \Tr(M(\tilde{\Lambda}_A\otimes\Lambda_{\rm B})[\rho_{\rm AB}^\Gamma])-1\bigg\}\\
	&=\frac{1}{2}\bigg\{\max_{-\openone\leq M\leq \openone} \Tr((\tilde{\Lambda}_A\otimes\Lambda_{\rm B})^\dagger[M]\rho_{\rm AB}^\Gamma)-1\bigg\}\\
	&=\frac{1}{2}\bigg\{\|\tilde{\Lambda}_A\otimes\Lambda_{\rm B}\|^H_1\max_{-\openone\leq M\leq \openone} \Tr\bigg(\frac{(\tilde{\Lambda}_A\otimes\Lambda_{\rm B})^\dagger[M]}{\|\tilde{\Lambda}_A\otimes\Lambda_{\rm B}\|^H_1}\rho_{\rm AB}^\Gamma\bigg)\\
	&\qquad\quad-1\bigg\}\\
	&\leq \frac{1}{2}\bigg\{\|\tilde{\Lambda}_A\otimes\Lambda_{\rm B}\|^H_1\max_{-\openone\leq M'\leq \openone} \Tr(M'\rho_{\rm AB}^\Gamma)-1\bigg\}\\
	&= \frac{1}{2}\big\{\|\tilde{\Lambda}_A\otimes\Lambda_{\rm B}\|^H_1\|\rho_{\rm AB}^\Gamma\|-1\big\}\\
	&=\|\tilde{\Lambda}_A\otimes\Lambda_{\rm B}\|^H_1N(\rho_{\rm AB})+\frac{\|\tilde{\Lambda}_A\otimes\Lambda_{\rm B}\|^H_1-1}{2}.
\end{aligned}
\end{multline}

Solving for $N(\rho_{\rm AB})$ we finally find
\beq
\begin{aligned}
N(\rho_{\rm AB}) \geq \frac{1}{\|\tilde{\Lambda}_A\otimes\Lambda_{\rm B}\|^H_1}&\bigg(N((\Lambda_{\rm A}\otimes\Lambda_{\rm B})[\rho_{\rm AB}])\\
&-\frac{\|\tilde{\Lambda}_A\otimes\Lambda_{\rm B}\|^H_1-1}{2}\bigg).
\end{aligned}
\eeq

For a Hermiticity preserving map $\Theta$ one easily checks that $\|\Omega\circ\Theta\|^H_1\leq \|\Omega\|^H_1\|\Theta\|^H_1$. In our case
\beq
\begin{aligned}
\|\tilde{\Lambda}_A\otimes\Lambda_{\rm B}\|^H_1&\leq \|\tilde{\Lambda}_A\otimes\idmap_{\rm B_{out}}\|^H_1\|\idmap_{\rm A_{in}}\otimes\Lambda_{\rm B}\|^H_1\\
	&=\|{\Lambda}_A\otimes\idmap_{\rm B_{out}}\|^H_1\|\idmap_{\rm A_{in}}\otimes\Lambda_{\rm B}\|^H_1\\
	&\leq \|\Lambda_A\|_\diamond\|\Lambda_B\|_\diamond.
 \end{aligned}
\eeq
{By $\idmap_{\rm B_{out}}$ and $\idmap_{\rm A_{in}}$ we have denoted the identity map on the output space of $\Lambda_{\rm B}$ and on the input space of  $\Lambda_{\rm A}$, respectively. The equality in the second line is due to the fact that  $\|T\circ\Omega\circ T\|^H_1=\|\Omega\|^H_1$. The diamond norm~\cite{diamondnorm} is defined as  $\|\Omega\|_\diamond\equiv\|\Omega\otimes\idmap\|_1$, with the identity map that can be taken as acting on the same input space as $\Omega$,  and with $\|\Omega\|_1\equiv\sup_{\|X\|_1\leq 1}\|\Omega[X]\|_1$ the trace norm for maps. The last inequality is due to the fact that $\|\Omega\otimes\idmap\|_1^H\leq\|\Omega\otimes\idmap\|_1\leq \|\Omega\|_\diamond$, for $\idmap$ acting on an arbitrary dimension~\cite{diamondnorm}}.
Thus, we finally obtain the lower bound
\beq
\begin{aligned}
N(\rho_{\rm AB})
\geq 
\frac{1}{\|\Lambda_A\|_\diamond\|\Lambda_B\|_\diamond}&\bigg(N((\Lambda_{\rm A}\otimes\Lambda_{\rm B})[\rho_{\rm AB}])\\
&-\frac{\|\Lambda_A\|_\diamond\|\Lambda_B\|_\diamond-1}{2}\bigg).
\end{aligned}
\eeq
\mbox{ }

\bibliographystyle{apsrev}

\end{document}